\documentclass{article}     
\usepackage{macros}
\usepackage{authblk}
\usepackage{fullpage}

 \title{Attenuate Locally, Win Globally: An Attenuation-based Framework for Online Stochastic Matching with Timeouts \thanks{This is the full version of the paper that appeared in AAMAS-2017 \cite{AAMAS17}. There was an error in one of the offline black-box. This version fixes all the ratios in the theorems. Research supported in part by NSF Awards CNS-1010789 and CCF-1422569, by a gift from Google, Inc., and by research awards from Adobe, Inc.\ and Amazon Inc.}}

\author{Brian Brubach\thanks{\textbf{Email: } \texttt{bbrubach@cs.umd.edu}}}
\author{Karthik A. Sankararaman\thanks{\textbf{Email: }\texttt{kabinav@cs.umd.edu}}}
\author{Aravind Srinivasan\thanks{\textbf{Email: } \texttt{srin@cs.umd.edu}}}
\author{Pan Xu\thanks{\textbf{Email: } \texttt{panxu@cs.umd.edu}}}
\affil{Department of Computer Science, University of Maryland, College Park, MD 20742, USA}

\date{First version: September 2017\\ This version: \today}

\begin{document}

\maketitle

\begin{abstract}

%Online matching problems have garnered significant attention in recent years due to numerous applications (see book by Aranyak Mehta). Many of them incorporate stochasticity to model uncertainty in the real world. The \textit{Online Stochastic Matching with Timeouts} problem introduced by Bansal, Gupta, Li, Mestre, Nagarajan, and Rudra (\emph{Algorithmica}, 2012) includes stochasticity in both the arrival process and the matching process. The online arrivals are drawn independent and identically distributed (i.i.d.) from a known distribution. Additinally, each edge in the graph is assigned a probability of existing (independent of all other edges). Edges must be probed to determine if they exist before being matched. Each online vertex has a \textit{timeout} limit on the number of adjacent edges that can be probed.

Online matching problems have garnered significant attention in recent years due to numerous applications in e-commerce, online advertisements, ride-sharing, etc. Many of them capture the uncertainty in the real world by including stochasticity in both the arrival and matching processes. The Online Stochastic Matching with Timeouts problem introduced by Bansal, \emph{et al.,} (\emph{Algorithmica}, 2012) models matching markets (e.g., E-Bay, Amazon). Buyers arrive from an independent and identically distributed (i.i.d.) known distribution on buyer profiles and can be shown a list of items one at a time. Each buyer has some probability of purchasing each item and a limit (timeout) on the number of items they can be shown. 

Bansal \emph{et al.,} (\emph{Algorithmica}, 2012) gave a $0.12$-competitive algorithm which was improved by Adamczyk, et al., (\emph{ESA}, 2015) to $0.24$. We present several online attenuation frameworks that use an algorithm for offline stochastic matching as a black box. On the upper bound side, we show that one framework, combined with a black-box adapted from Bansal \emph{et al.,} (\emph{Algorithmica}, 2012), yields an online algorithm which nearly doubles the ratio to $0.46$. Additionally, our attenuation frameworks extend to the more general setting of fractional arrival rates for online vertices. On the lower bound side, we show that no algorithm can achieve a ratio better than $0.632$ using the standard LP for this problem. This framework has a high potential for further improvements since new algorithms for offline stochastic matching can directly improve the ratio for the online problem. 

Our online frameworks also have the potential for a variety of extensions. For example, we introduce a natural generalization: Online Stochastic Matching with Two-sided Timeouts in which both online and offline vertices have timeouts. Our frameworks provide the first algorithm for this problem achieving a ratio of 0.30. We once again use the algorithm of Bansal \emph{et al.,} (\emph{Algorithmica}, 2012) as a black-box and plug it into one of our frameworks.

\end{abstract}

\setcounter{page}{1}
\pagenumbering{arabic}

\section{Introduction}

Consider a typical problem in two-sided matching markets (\eg E-Bay, Amazon). We have a certain number of buyer profiles and items. Let us denote the set of buyer profiles by $V$ and the set of items by $U$. In our initial problem, we assume that an item is present in the market until bought. We have an $n$-round process. In each round, a buyer, sampled uniformly at random (with replacement) from the buyer profiles, arrives (this is equivalent to assuming a Poisson arrival process on the buyers when $n$ is large). We assume that we have $n$ buyer profiles (this assumption is called \textit{integral arrival rates} in the literature \cite{mehtaBook}). For every item $u \in U$ and buyer $v \in V$, let $e$ denote $(u,v)$ and $p_e$ denote the probability that the buyer $v$ will buy the item $u$. If $v$ buys $u$, then we obtain a reward of $w_e$. We assume all of these buying events are independent. Since every buyer $v$ has a limited attention-span, they can be shown at most $t_v$ items where $t_v$ is typically called a \textit{timeout} or \textit{patience} for $v$. 
When a buyer arrives, we show them items one-by-one until either they choose to buy one item or they have seen $t_v$ items, whichever comes first. The goal is to design an algorithm such that the expected reward at the end of the $n$-round process is maximized. Usually, the buyer profiles are gathered based on historical information. Hence, we will assume that the system knows all of the buyer profiles as well as the buying probabilities for each item-buyer pair. 

We model this as the \textit{Online Stochastic Matching with Timeouts} (\osmt) problem introduced by Bansal, Gupta, Li, Mestre, Nagarajan, and Rudra \cite{BansalLPCure}. Formally, this is a \emph{probe-commit} model for online bipartite stochastic matching. We are given a bipartite graph $G = (U, V, E)$ as input. Let us represent an item-buyer pair by an edge $e = (u, v)$ with a probability $p_e$ of existing (independent of other edges) and a weight $w_e$. The probability $p_e$ models a buyer's interest in an item and the weight $w_e$ is the profit obtained by successfully matching $v$ to $u$. Each vertex $v \in V$ has a timeout $t_v$, which is the number of times we may probe one of its edges; however, the vertices in $U$ have no timeout restrictions (equivalently, we can say they have timeouts of infinity). The algorithm can ``probe'' $e$ to see if the buyer $v$ ``buys'' the item $u$. If they do, the decision of selling $u$ to $v$ is made irrevocably (commit). Otherwise, if they do not buy $u$, we may probe another edge incident to $v$ provided their patience $t_v$ has not been exhausted. 
%These values are all known \emph{a priori}. 
The values $p_e$, $w_e$, and $t_v$ for the graph $G$ are known \emph{a priori}.

The algorithm proceeds for $n$ rounds. In each round, a vertex $v$ arrives and we can probe at most $t_v$ neighbors in an attempt to match $v$. Arrivals are drawn with replacement from a known i.i.d. distribution on $V$. For simplicity, we will consider the uniform distribution.\footnote{Our results are applicable to any distribution with integral arrival rates (the expected number of arrivals for each online vertex $v$ is an integer $r_v$ with $\sum_v r_v = n$).} 
If a probed edge $(u, v)$ is found to exist, we must match $v$ to $u$ and no more probing is allowed for that round. Each vertex in $U$ can be matched at most once. The vertices in $V$ are called \textit{types} (a buyer profile) and two or more arrivals of the same type $v \in V$ are considered distinct vertices (two different buyers of a particular profile) which can each be probed up to $t_v$ times and matched to separate neighbors in $U$. The objective is to maximize the expected weight (or profit) of the final matching obtained. Unsurprisingly, this model captures problems beyond the buyer/seller scenario described above. Various online stochastic matching problems have been considered for online advertising and many other applications \cite{mehtaBook}.

We give new algorithms for the Online Stochastic Matching with Timeouts problem that improve the competitive ratio over the previous work. We then introduce a \textit{new model} wherein the seller selling item $u$ has a limited patience as well. In this generalization, called \textit{Online Stochastic Matching with Two-sided Timeouts}, we have an additional constraint that every vertex $u \in U$ has a timeout $t_u$ and the algorithm can probe at most $t_u$ neighbors of $u$ across the $n$ rounds. 
%If the vertex is not matched after $t_u$ probes, it disappears from the system. 
We give the first constant-factor approximation algorithm for this generalized setting.

\xhdr{Related Work.}
		Online bipartite matching and its variants have been an active area of study beginning with the seminal work of Karp, Vazirani, and Vazirani \cite{kvv}. They studied unweighted bipartite matching in the adversarial arrival model and gave an optimal $(1 - 1/e)$ competitive algorithm. The advent of e-commerce and ad-allocation brought more variants of this problem. Many variants, which we will not summarize here, study arrivals in a random or adversarial order on an unknown set of online vertices. For an exhaustive literature survey, we refer the reader to the book of Aranyak Mehta \cite{mehtaBook}. 
		
		In the i.i.d. arrival model, Feldman \emph{et al.} \cite{bib:Feldman}, Bahmani and Kapralov \cite{bahmani2010improved}, Manshadi \emph{et al.} \cite{bib:Manshadi}, Haeupler \emph{et al.} \cite{bib:Haeupler}, Jaillet and Lu \cite{bib:Jaillet}, and Brubach \emph{et al.} \cite{ESA16} gave improved algorithms for the Online Stochastic Matching problem. The term ``stochastic'' here refers to the known i.i.d. arrival model, although some of those papers also address stochastic edge models. The Online Stochastic Matching problem can be seen as a special case of Online Stochastic Matching with Timeouts wherein all arriving vertices $v$ (buyers) have patience $t_v = 1$.
	
%	Beyond online matching, various generalizations of this problems have been studied. The adwords problem was first introduced by Mehta \emph{et al.} \cite{Mehta2007Adwords} and subsequently studied by Buchbinder \emph{et al.} \cite{Buchbinder2007} and Devanur and Hayes \cite{devanur2009adwords}. Further generalizations were considered as resource allocation problem by Devanur \emph{et al.} \cite{DevanurJainSivanWilkens} and Devanur \emph{et al.} \cite{DevanurSivanAzar}. Another version was considered by Devanur and Jain \cite{devanurjainconcave} where they generalized the notion of maximum budget to arbitrary concave functions. Other generalizations that capture the online matching problem are Online Packing Linear Programs by Feldman \emph{et al.} \cite{FeldmanHenzinger} and Agrawal \emph{et al.} \cite{AgrawalWangYe} and the study of Online Convex Programs by Agrawal and Devanur \cite{AgrawalDevanur}. These generalizations are again outside the scope of this paper.
	
	Beyond online matching, other related problems have been studied. The Adwords problem was introduced by Mehta \emph{et al.} \cite{Mehta2007Adwords} and subsequently studied by Buchbinder \emph{et al.} \cite{Buchbinder2007} and Devanur and Hayes \cite{devanur2009adwords}. More variants have been considered by Devanur \emph{et al.} \cite{DevanurJainSivanWilkens}, Devanur \emph{et al.} \cite{DevanurSivanAzar}, and Devanur and Jain \cite{devanurjainconcave}. Other generalizations that capture the online matching problem are Online Packing Linear Programs by Feldman \emph{et al.} \cite{FeldmanHenzinger} and Agrawal \emph{et al.} \cite{AgrawalWangYe} and the study of Online Convex Programs by Agrawal and Devanur \cite{AgrawalDevanur}.

%	The work that most directly relates to our work is that of Bansal \emph{et al.} \cite{BansalLPCure}. This work introduced the problem of Online Stochastic Matching with Timeouts and gave the first constant factor approximation of $0.12$. This factor was later improved to $0.24$ by Adamczyk \emph{et al.} \cite{AGM}. In both of these works, they consider the notion of timeouts on only the online vertices. The original motivation for timeouts came from the patience constraints in the Offline Stochastic Matching problem. This problem was first introduced by Chen \emph{et al.} \cite{Chen} and later studied by Bansal \emph{et al.} \cite{BansalLPCure}, Adamczyk \emph{et al.} \cite{AGM}, and Baveja \emph{et al.} \cite{baveja2015}. A generalization of this to packing problems was studied by Gupta and Nagarajan \cite{Gupta}. Other models of Offline Stochastic Matching problems have been studied in the literature which is not in the scope of this paper.

	Bansal \emph{et al.} \cite{BansalLPCure} introduced the problem of Online Stochastic Matching with Timeouts and gave the first constant factor competitive ratio of $0.12$. This was later improved to $0.24$ by Adamczyk \emph{et al.} \cite{AGMesa}. In both works, they considered the notion of timeouts only on the online vertices. The original motivation for timeouts came from the patience constraints in the Offline Stochastic Matching problem. This offline problem was first introduced by Chen \emph{et al.} \cite{Chen} and later studied by Bansal \emph{et al.} \cite{BansalLPCure}, Adamczyk \emph{et al.} \cite{AGMesa}, and Baveja \emph{et al.} \cite{baveja2015}. A generalization to packing problems was studied by Gupta and Nagarajan \cite{Gupta}.

%	The $b$-matching problem is yet another variant of the Online Matching problem that has been previously studied. Among all other generalizations, this problem relates most closely with the variant of the problem we introduce, in which we have a patience constraint on both the online and the offline vertices. Similar to how $b$-matching models the fact that a particular ad can be shown multiple times, patience on the left hand side models this in the stochastic rewards setting.
%	In the adversarial setting, $b$-matching was first studied by Kalyanasundaram and Pruhs \cite{Kalyanasundaram} where they gave an optimal algorithm. Alaei \emph{et al.} \cite{Alaei} studied the prophet-inequality problem and consider the stochastic i.i.d. setting. They give an algorithm whose competitive ratio is $1-\frac{1}{\sqrt{b+3}}$. Alaei \emph{et al.} \cite{Alaei2013} and Brubach \emph{et al.} \cite{ESA16} study the Online b-matching Stochastic Matching problem under the i.i.d. setting and give a ratio of $1- O(1/\sqrt{b})$. The particular generalization studied in this paper, has no direct previous work.
	
	The Online Stochastic Matching with Two-sided Timeouts problem which we introduce here has no direct previous work. One related problem is Online $b$-matching wherein the offline vertices can each be matched at most $b$ times. This is somewhat similar to having timeouts on the offline vertices in online stochastic matching. In the adversarial setting, Online $b$-matching was first studied by Kalyanasundaram and Pruhs \cite{Kalyanasundaram} and they gave an optimal algorithm. Alaei \emph{et al.} \cite{Alaei} studied the prophet-inequality problem and considered the stochastic i.i.d. setting. They gave an algorithm whose competitive ratio is $1-\frac{1}{\sqrt{b+3}}$. Alaei \emph{et al.} \cite{Alaei2013} studied the Online b-matching Stochastic Matching problem in the i.i.d. setting and gave a ratio of $1- O(1/\sqrt{b})$.

	 % subsection of intro

\subsection{Preliminaries}
\label{sec:prelim}

%We now describe the theoretical preliminaries and the required background for this paper. 
Henceforth, we refer to the items $U$ as \textit{offline vertices} and the buyers $V$ as \textit{online vertices}. We refer to the reward $w_e$ as the \textit{weight} of the edge while we refer to the buying probability $p_e$ as the \textit{edge probability}. 
We use the term \textit{round} to refer to the arrival of a single online vertex and our attempts to match it. Thus, our algorithms will proceed over $n$ rounds. 
When we refer to \textit{time} $t \in [n]\doteq\{1,2,\ldots, n\}$, we refer to the beginning of the $t^{th}$ round in the process.  \bluee{In this paper, we consider a more general setting than that in \cite{BansalLPCure} and \cite{AGMesa}  where each online vertex $v$ has an arbitrary fractional arrival rate $r_v$ with $\sum_v r_v=n$. Additionally, we assume that $r_v \in (0,1]$ for every $v \in V$. This assumption is without loss of generality (WLOG), since we can create multiple copies of a vertex $v$ if $r_v > 1$ to satisfy the above condition.} We say that a vertex in the offline set $U$ is \textit{safe} at some time $t$ if it has not been matched in a previous round. Similarly, an edge $e=(u,v)$ is {\it safe} in round $i$ iff $u$ is available in an arrival of $v$ at round $i$. A safe edge is \textit{safely probed} if it is probed before $v$ is matched or the timeout $t_v$ is reached. Finally, $\partial(u)$ and $\partial(e)$ denote the edges incident to vertex $u$ and edge $e$, respectively. \bluee{We summarize all of the notation in Section~\ref{appx:notations} of the Appendix.}

The competitive ratio for this class of problems is defined slightly differently from most other online algorithms (see \bluee{Section 2.2.1 in \cite{mehtaBook}}). For a given instance $\mathcal{I}$, let $\mathbb{E}[\ALG(\mathcal{I})]$ denote the expected reward obtained by the algorithm on this instance and let $\mathbb{E}[\OPT(\mathcal{I})]$ denote the expected value of the offline optimal solution. 
%{\color{red}We note that the offline optimal algorithm knows which vertices have arrived and can probe edges/vertices in any order, but does not know the outcomes of the edge probabilities. } 
The competitive ratio can be defined as $\min_{\mathcal{I}} \mathbb{E}[\ALG(\mathcal{I})]/\mathbb{E}[\OPT(\mathcal{I})]$. A commonly used technique in the literature is to construct an appropriate linear program, called the benchmark LP, whose optimal value upper bounds $\mathbb{E}[\OPT(\mathcal{I})]$. Hence, comparing the expected reward obtained by the algorithm to the optimal value of this LP immediately leads to a lower bound on the competitive ratio. For completeness, we define the offline problem as follows. We are given the full graph realized after $n$ rounds of sampling vertices from the online set independently with replacement and we may probe edges in any order. Thus, the expectation of the offline OPT is taken over the random arrivals (and hence the graph), the random outcomes of each probe, and any randomness used by some optimal algorithm.

% When an online vertex $v$ arrives, we use the notation $G(v)$ to refer to the star graph of $v$ and its safe neighbors in $U$.

\xhdr{Benchmark Linear Program (LP).}
We use the following Linear Program as a benchmark for our competitive ratios. 

\vspace*{-0.1in}
\begin{alignat}{2}
\label{lp2:stoch-match}
\text{maximize}    & \sum_{e \in E} w_{e} f_e p_e     \\
\text{subject to}  
	& \sum_{e \in \partial(u)} f_e p_{e} \leq 1	&\ & \quad\forall u \in U \label{lp-cons:umatch}\\
	& \sum_{e \in \partial(v)} f_e p_e \leq r_v	&\ & \quad\forall v \in V \label{lp-cons:vmatch}\\
			& \sum_{e \in \partial(u)} f_e  \leq t_u	&\ & \quad\forall u \in U \label{lp-cons:upatience}\\
		&\sum_{e \in \partial(v)} f_e  \leq t_v\cdot r_v	&\ & \quad\forall v \in V \label{lp-cons:vpatience}\\
		& 0 \le f_e \le r_v & \ &  \quad\forall e \in E
\end{alignat} 

\bluee{This linear program was introduced in \cite{BansalLPCure}. When the rates $r_v=1$ for every $v \in V$ and $t_u=n$ for every $u \in U$, Lemma 11 in \cite{BansalLPCure} shows that the optimal value of this LP upper-bounds the optimal value of the best online algorithm. A similar proof extends to the case when the rates $r_v$ are in $[0, 1]$ and the time-outs $t_u$ are any integers in $[n]$.}

The variable $f_e$ in the above LP refers to the expected number of probes on $e$ in the offline optimal. Constraints~\eqref{lp-cons:umatch} and \eqref{lp-cons:vmatch} represent the matching constraint in the offline graph and Constraints~\eqref{lp-cons:upatience} and \eqref{lp-cons:vpatience} represent patience constraints (timeouts) on the seller and buyer side, respectively. \bluee{When we have no timeouts on the offline set, we set $t_u = n$, for all $u \in U$. Informally, Constraint \eqref{lp-cons:vmatch} captures the fact that the expected number of arrivals of $v$ over the $n$ online rounds is equal to $r_v$. Thus, the expected number of matches for $v$ is no more than $r_v$. }

\xhdr{Overview of Attenuation Frameworks.}
\label{sec:overview}
	We present several online \textit{attenuation frameworks} for the design and analysis of algorithms for the Online Stochastic Matching with Timeouts problem. Informally, an attenuation framework is a method to balance the performance of every edge across all of the rounds thus improving the worst case performance (over all edges). We analyze the expected performance for every edge individually across the $n$ rounds and use the linearity of expectation to compute the total expected reward. It is easy to see that a lower bound on the competitive ratio can be obtained by considering the edge $e=(u,v)$ with the lowest ratio. However, it is common that the neighboring edges $e' \in \partial(e)$ of $e$ have much higher ratios. Moreover, the performance of $e$ is negatively affected by that of $\partial(e)$. Thus, we can improve the performance of $e$ by attenuating its neighboring edges $e'\in \partial(e)$ and in turn, improve the overall competitive ratio.

When a vertex $v$ arrives in each online round, we need to probe at most $t_v$ of its neighbors sequentially until either $v$ is matched, $v$ has no more safe neighbors, or we have made $t_v$ unsuccessful probes. This is essentially the offline stochastic matching problem studied in~\cite{BansalLPCure} on a star graph, $G(v)$, which consists of $v$ and its safe neighbors. Hence we can view our setting as follows. Suppose we have a \emph{black box}, which is an LP-based algorithm solving an offline stochastic matching problem with timeouts on a general star graph. Then any online framework takes as input a black box with a designated property and outputs an online algorithm. The final competitive ratio is determined jointly by the online framework and the input black box (see Theorems \ref{thm:attn1}, \ref{thm:attn2}, \ref{thm:attn3}, \ref{thm:ext}). 
%We will see that any future improvements in the {\it black box} can be easily paired with the best possible framework to get an improved result. 

%Each attenuation framework takes a \textit{black box} as an input. A black box  $\BB$ can be understood as an LP-based offline algorithm applied to a star graph $G(v)$: $v$ is the arrival vertex and $G(v)$ consists of its safe neighbors; for any feasible solution $\g$ satisfying the patience and matching constraints on $G_v$, $\BB[\g]$ outputs a feasible probing strategy on $G_v$. More details can be referred to Section [XX]. 
	
	%Each attenuation framework requires a \textit{probing strategy} which, given an arriving vertex $v$ and an LP solution, probes at most $t_v$ neighbors of $v$. 	%selects at most $t_v$ edges and a probing order. 
	%This probing strategy is treated as a black box with the only restriction being that it must satisfy a single property designated by the chosen attenuation framework. Thus, any future improvements in probing strategies can be easily paired with the best possible framework to get an improved result. 
% Added text to promote second black box

	The idea of \textit{edge-attenuation}, first proposed in~\cite{AGMesa}, aims to balance the performance of all edges. Suppose our black box guarantees that each edge $e$ will be probed with probability \textit{at least} $\alpha f_e$, where $\alpha$ is some constant and $f_e$ is the value assigned by an LP relaxation. Edge-attenuation will guarantee each edge is probed with probability \textit{equal to} $\alpha f_e$. If we know that the black box probes $e$ with probability $\alpha' f_e > \alpha f_e$, we can achieve our goal by setting a probability $1 - \alpha/\alpha'$ with which we ``pretend'' to probe $e$, but don't actually probe it.
	Unfortunately, the exact value of $\alpha'$ is hard to find since it is jointly determined by multiple sources of  randomness in the algorithm and the input. Therefore, we resort to a Monte-Carlo based simulation for a sharp estimate, e.g., simulating the algorithm on the input instance many times and taking the sample mean as an estimate of $\alpha'$. Simulation-based attenuation has been previously used in the stochastic knapsack problem \cite{ma2014improvements} and the offline stochastic matching problem \cite{AGMesa}. As shown in those works, we can ensure that the simulation errors add up to give at most an additive factor of $\epsilon$ in the final ratio (for any desired positive constant $\epsilon > 0$). 

	Our novel \textit{vertex-attenuation} approach applies a similar idea to all of the edges incident to an offline vertex. Notice that the star graph centered on an arriving vertex $v$ will have a smaller number of incident edges in later rounds as the offline neighbors of $v$ have been matched in earlier rounds. Intuitively, each edge in this smaller graph has less competition and therefore a greater chance to be probed. Since the probability of being matched in a previous round can differ significantly among the offline vertices, we apply attenuation to the offline vertices to bound these probabilities in our analysis.

\xhdr{Our Contributions.}
One of our main contributions is a general template for solving the Online Stochastic Matching with Timeouts problem which gives rise to a class of algorithms. Notably, we decouple the offline subproblem of which edges to probe when a vertex arrives from the online problem of handling a series of arrivals.

The offline subproblem addressed in Section~\ref{sec:obb} takes as input a stochastic star graph and an LP solution on that graph. The output is a probing strategy which preserves the LP values up to some factor in expectation. Our template allows the offline subproblem to be solved by any black box algorithm provided it satisfies one of the three basic properties described in Section~\ref{sec:properties}. Thus any new algorithm for this subproblem can easily be plugged into our overall algorithm and the analysis. For the purposes of this paper, we plug-in the prior work of \cite{BansalLPCure} as the ``black-box'' and obtain improved ratios.

For the online framework (Section~\ref{sec:frame}), we bring tighter, cleaner analysis to the edge-attenuation approach in~\cite{AGMesa} and generalize it to work with a broad class of algorithms for the offline subproblem. For example,~\cite{AGMesa} uses only the first arrival of each online vertex type, discarding subsequent arrivals of the same type while we use every arrival. We then present a new vertex-attenuation approach which can be combined with edge-attenuation to achieve further improvement in the competitive ratio. The previous algorithm of \cite{AGMesa} gave a ratio of $0.24$, while we almost double this to give a ratio of $0.46$. 

%Additionally, as will see later the algorithm by \cite{AGM} is a special case of our framework. Hence, we are able to give a cleaner analysis to their result as well.

We introduce the more general \textit{Online Stochastic Matching with Two-sided Timeouts} problem in Section~\ref{sec:extension}. This new problem is well-motivated from applications and theoretically interesting. Using our template and one of the attenuation frameworks, we give a constant factor $0.30$-competitive algorithm for this problem.

Finally, we show in Section~\ref{sec:lb} that no algorithm using the LP in this paper and in prior work can achieve a ratio better than $1 - 1/e \approx 0.632$.

 \section{Offline Black Box}
\label{sec:obb}

%and $\delta(v)$ is the set of safe neighbors.
The online process consists of $n$ offline rounds. In each round, we have an offline stochastic matching instance as studied in \cite{BansalLPCure,AGMesa} on a star graph. Consider a single round at time $t$. Let $v$ be the arriving vertex and $G(v)$ be the star graph of $v$ and its safe neighbors. For ease of notation, we overload $G(v)$ to denote both the set of edges and the set of neighbors of $v$ \footnote{Since we deal with star graphs, the two sets have the same cardinality.}. 

\bluee{Consider the following polytope.
\begin{equation} \label{eqn:LP-off}
\left\{
\sum_{e \in G(v)} g_{e} p_{e} \leq 1, ~~ 
\sum_{e \in G(v)} g_{e}  \leq t_v , ~~ 
0 \le g_e \le 1 , \forall e \in G(v)
\right\}
\end{equation}}

An offline {\it black box} is any algorithm that transfers a \emph{feasible} solution to the polytope \eqref{eqn:LP-off} to a feasible probing strategy on $G(v)$ with a guaranteed performance for each edge. By a \textit{feasible} probing strategy, we mean one that does not violate the matching and patience constraints on $v$. Now we present one concrete example of a black box which was known in prior work (\cite{BansalLPCure}).  Let $\tilde{\mathbf{g}}=\{\tilde{g}_e| e\in G(v)\}$ be any given feasible solution to the polytope \eqref{eqn:LP-off}. We will use $\tilde{g}_e$ to refer to the value of $\tilde{\mathbf{g}}$ for edge $e$.

\xhdr{Uniform Random Black Box.}\label{sec:bb1}
The Uniform Random Black Box, denoted by $\BB_{UR}$, is a direct application of the algorithm in \cite{BansalLPCure} to the star graph $G(v)$. To be consistent with the notation in~\cite{BansalLPCure}, we use GKPS to denote the dependent rounding techniques developed in Gandhi {\it et al}~\cite{bib:Gandhi}. \bluee{Critically, we use properties (P1)-(P3) in \cite{bib:Gandhi}, which we restate here. The rounding procedure takes a fractional solution $\tilde{\mathbf{g}}$ as input and produces an integer vector $\mathbf{\hat{G}}$ that satisfy the following properties.

\begin{enumerate}
	\item \xhdr{Marginal Distribution.} For every edge $e=(u, v)$, we have $\Pr[\hat{G}((u, v)) = 1] = \tilde{g}_e$.
	\item \xhdr{Degree-preservation.} For any vertex $u \in U$, the number of edges incident to $u$ and present in $\mathbf{\hat{G}}$ is exactly equal to $t_u$.
	\item \xhdr{Negative correlation.} For any vertex $u \in U$ and any subset $S$ of edges incident to $u$, we have,
		\[ \forall b \in \{0, 1\} \quad \Pr[\wedge_{e=(u, v') \in S} (\hat{G}((u, v')) = b)] \leq \prod_{e=(u, v') \in S} \Pr[(\hat{G}((u, v')) = b] \]
\end{enumerate}
}

\begin{algorithm}[!h]
		\label{alg:BBUR}
		\caption{$\mathsf{BB}_{UR}$}
		\DontPrintSemicolon
Apply GKPS \cite{bib:Gandhi} to $G(v)$  where edge $e$ is associated with a value $g_e$. Let $\hat{G}(v)$ be the set of edges that gets rounded;

Choose a random permutation $\pi$ over $\hat{G}(v)$. Probe each edge $ e \in \hat{G}(v)$ in the order $\pi$ until $v$ is matched. 
	\end{algorithm}
 
 The performance of $\BB_{UR}$ is presented in Lemma \ref{lem:blackbox}. \bluee{For each given $e \in G(v)$, let $\lam(e, \tilde{\mathbf{g}})=\sum_{e'\neq e} \tilde{g}_{e'} p_{e'}$}.

 \begin{lemma}
 	\label{lem:blackbox}
 In $\BB_{UR}$, each edge $e$ is probed with probability at most $\tilde{g}_e$ and at least $\left(1-\frac{\lam(e,\tilde{\mathbf{g}})}{2}\right) \tilde{g}_e$, where $\lam(e, \tilde{\mathbf{g}})=\sum_{e'\neq e} \tilde{g}_{e'} p_{e'}$.
 \end{lemma} 

\begin{proof}
	Each edge $e$ can be probed only when $e$ appears in $\hat{G}_v$ which occurs with probability $\tilde{g}_e$, according to GKPS \cite{bib:Gandhi}. Therefore, we have that each edge $e$ is probed  with probability at most $\tilde{g}_e$.
	
Consider an edge $e$. The random permutation order $\pi$ can be viewed as follows: each $e'$ uniformly draws a real number $a_{e'}$ from $[0,1]$ and we sort all $\{a_{e'}\}$ in increasing order. Condition on $a_e=x \in [0,1]$ and for each $e' \neq e$, let $X_{e'}$ be the indicator random variable that is $1$ if $a_{e'}<x$ (i.e., $e'$ falls before $e$ in $\pi$). For each $e'$, let $Y_{e'}$ and $Z_{e'}$ be the respective indicator random variables corresponding to the event $e'$ being rounded and if $e'$ is present when probed. \bluee{Therefore we have,
	\begin{align}
		\Pr[\mbox{$e$ is probed} \given{a_e=x, Y_e=1}] 		& \ge \E\left[1-\sum_{e'\neq e}X_{e'}Y_{e'} Z_{e'} \given{a_e=x, Y_e=1}\right] \label{ineq:bb-1} \\
	&=1- \sum_{e'\neq e} x p_{e'} \E[Y_{e'}=1 | Y_e=1] 
 \\		
		& \ge 1-x \sum_{e' \neq e} \tilde{g}_{e'}p_{e'}=1-x \lam(e,\tilde{\mathbf{g}})	\label{ineq:bb-2}
		\end{align}
		
Inequality \eqref{ineq:bb-1} follows from the observation that when conditioned on the event that $e$ is rounded up (\ie $Y_e=1$), it is probed iff none of its neighboring edges falls before $e$ in the random permutation $\pi$ (\ie $X_{e'}=1$), is probed (\ie $Y_{e'}=1$) and is present (\ie $Z_{e'}=1$). Inequality \eqref{ineq:bb-2} is valid since $\E[Y_{e'}=1 | Y_e=1] \le \E[Y_{e'}=1 ]=\tilde{g}_{e'}$ (events that edges on the star graph are rounded are negatively correlated due to GKPS \cite{bib:Gandhi}).}

	Thus we have the following.
	\begin{align*}
\Pr[\mbox{$e$ is probed}]&=\tilde{g}_e \int_0^1 \Pr[\mbox{$e$ is probed}\given{a_e=x, Y_e=1}] dx \\
	& \ge \tilde{g}_e(1- \lam(e,\tilde{\mathbf{g}})/2) \qedhere
	\end{align*}
\end{proof}

\section{Attenuation Frameworks for Online Matching with Timeouts}
\label{sec:frame}

The main idea in all our attenuation frameworks is to decouple the offline and online subproblems. The offline problem concerns with an action on an arriving vertex (\emph{i.e.,} when a vertex arrives, we must decide which incident edges to probe). This is handled by a black box offline algorithm. The black box is only restricted by one of the three properties listed below. The online problem deals with managing a series of arrivals and that is the primary focus of this section.

Throughout this section, we assume that through Monte-Carlo simulations (referred to as simulations in short, henceforth) we can get an accurate estimate of our target probabilities. As shown in \cite{AGMesa,ma2014improvements}, we can manipulate the simulation error appropriately such that we lose at most an additive factor of $\epsilon = o(1)$ in the final ratio. 

\xhdr{Properties of black box.}
\label{sec:properties}
Property A states that the black box $\BB$ is guaranteed to probe each edge with probability at least $\alpha \tilde{g}_e$ for some constant $\alpha\in (0,1)$. It gives a lower bound on the performance of each edge without any further restrictions on the black box. Formally it can be stated as follows.

\textbf{Property A:} \emph{For any feasible solution $\tilde{\mathbf{g}}$ to LP \eqref{eqn:LP-off}, $\BB$ outputs a feasible probing strategy $\BB[\tilde{\mathbf{g}}]$ such that every edge $e$ will be probed with probability at least $\alpha \tilde{g}_e$ for some constant $\alpha\in (0,1)$.} 

 Recall that $\lam(e,\tilde{\mathbf{g}})=\sum_{e' \neq  e} \tilde{g}_{e'}p_{e'}$ and this value expresses the amount of competition $e$ will face from its neighbors. Properties B and C add the restriction that the probability of probing a given edge will be a function of both $\tilde{g}_e$ and $\lam(e,\tilde{\mathbf{g}})$. This allows us to take advantage of the fact that $\lam(e,\tilde{\mathbf{g}})$ may decrease as the number of arrivals increases.
The conditions of \textit{non-increasing and convexity} on the function $\mathsf{R}_{\BB}$ are required to ensure that the offline ratio of $\BB$ can be used to bound the overall competitive ratio. The condition of \textit{finitely bounded first derivative} guarantees that the error accumulated from simulation is bounded.

\textbf{Property B:} \emph{For any feasible $\tilde{\mathbf{g}}$ to LP \eqref{eqn:LP-off}, $\BB$ outputs a feasible probing strategy $\BB[\tilde{\mathbf{g}}]$ such that each edge $e$ is probed with probability at least $\tilde{g}_e \mathsf{R}_{\BB}[\lam(e,\tilde{\mathbf{g}})]$, where $\mathsf{R}_{\BB}$ is a non-increasing and convex} function and has finitely bounded first derivative on $[0,1]$. 

Property C adds a further restriction that each edge is probed with probability at most $\tilde{g}_e$, which states as follows: 

\textbf{Property C:} \emph{For any feasible $\tilde{\mathbf{g}}$ to LP \eqref{eqn:LP-off}, $\BB$ outputs a feasible probing strategy $\BB[\tilde{\mathbf{g}}]$ such that every edge $e$ is probed with probability at most $\tilde{g}_e$ and at least $\tilde{g}_e \mathsf{R}_{\BB}[\lam(e,\tilde{\mathbf{g}})]$, where $\mathsf{R}_{\BB}$ is a non-increasing and convex} function and has finitely bounded first derivative on $[0,1]$. 

\begin{observation}\label{obs:pro}
For any $\BB$ satisfying \textbf{Property B} or \textbf{Property C} with  $\mathsf{R}_{\BB}$, we have $\RR_{\BB}[x] \le \RR_{\BB}[0] \le 1$ for all $x \in [0,1]$.
\end{observation}

The fact that $\RR_{\BB}[0] \le 1$ can be seen from this example: consider the graph $G(v)$ which has exactly one edge $e=(u,v)$. Clearly, $\tilde{g}_e=1$ is a feasible solution to LP \eqref{eqn:LP-off}. Then, $\BB[\tilde{\mathbf{g}}]$ will probe $e$ with probability at least $\RR_{\BB}[0] \tilde{g}_e=\RR_{\BB}[0] $, implying $\RR_{\BB}[0] \le 1$.

Note that the black box $\BB_{UR}$ satisfies all three properties.

\subsection{Attenuation}
\label{sec:attenuation}
Our black box properties give us lower bounds on the probability that an edge or vertex will be matched at any given time during the online phase. Attenuation allows us to make those bounds tight by reducing the performance of any edge or vertex which is exceeding the lower bound. The intuition is that weakening the over-performing edges will increase the performance of the lowest performing edges that provide the worst case competitive ratio.

We define three distinct attenuation frameworks: \textit{edge attenuation} which requires an offline black box satisfying \textbf{Property A}, \textit{vertex-attenuation} which requires a black box satisfying \textbf{Property C}, and \textit{edge and vertex-attenuation} which requires an offline black box satisfying \textbf{Property B}. The edge-attenuation framework generalizes and clarifies the edge-attenuation approach of~\cite{AGMesa}. We also give an improved result due to tighter analysis. Vertex-attenuation is a novel approach introduced in this paper that upper bounds the probability that a vertex in $U$ will be safe at time $t$. This lets us exploit the fact that the star graph $G(v)$ will be smaller in later rounds leading to a higher probability of safely probing each of the remaining edges. It can be combined with edge-attenuation to get the best known result for this problem.

Let $\ff=\{f_e | e\in E\}$ be an optimal solution to LP \eqref{lp2:stoch-match}. Let $v$ be the vertex arriving at  time $t \in [n]$ and $G_t(v)$ be the star graph consisting of $v$ and its safe neighbors at $t$. \bluee{Throughout this section,``at time $t$'' refers to the more precise statement \emph{``at the beginning of time (or round) $t$ before any online actions are performed''} unless explicitly stated otherwise. Let $\g_{t,v}=\{f_{e}/r_v | e \in G_t(v)\}$, with $g_e=f_e/r_v$ for each $e$ incident to $v$. From Constraints in LP \eqref{lp2:stoch-match}, it follows that $\g_{t,v}$ is a feasible solution to LP \eqref{eqn:LP-off} for the graph $G_t(v)$.

\xhdr{Edge-attenuation.} 
The most basic form of attenuation we consider is edge attenuation. Suppose we are given a black box $\BB$ satisfying \textbf{Property A} and guaranteeing that each edge is probed with probability at least $\alpha g_e$. This attenuation will guarantee that in each offline subproblem, each edge is probed with probability equal to $\alpha g_e$.

From \textbf{Property A}, we know that $\BB[\g_{t,v}]$ will probe each edge $e$ with probability at least $\alpha g_e$. {\it In this framework, we maintain that each safe edge $e$ is probed with probability exactly equal to $\alpha g_e=\alp f_e/r_v$ in all rounds via appropriate edge-attenuation}. Algorithm \ref{alg: the first attenu} gives a formal description of the algorithm.

\begin{algorithm}[!h]
		\label{alg: the first attenu}
		\caption{$\mathsf{ATTN}_{1}[\BB]$ }
		\DontPrintSemicolon
For each $t \in [n]$, let $v$ be the vertex arriving at time $t$ and $G_t(v)$ be the graph consisting of $v$ and its safe neighbors.\;
\bluee{Let $\g_{t,v}=\{f_e/r_v| e \in G_t(v)\}$ be the induced feasible solution to LP  \eqref{eqn:LP-off} from an optimal solution $\ff$ to LP \eqref{lp2:stoch-match}}.

Apply $\BB[\g_{t,v}]$ and simulation-based edge-attenuation (see Sec.~\ref{sec:overview}) to $G_t(v)$, such that each $e $ is probed with probability exactly equal to $\alpha g_e$.
	\end{algorithm}
	
\begin{theorem} \label{thm:attn1}	
For any $\BB$ satisfying \textbf{Property A}, $\mathsf{ATTN}_1[\BB]$ has a competitive ratio of $1-e^{-\alpha}$.
\end{theorem}

\begin{proof}
Consider an edge $e=(u,v)$ and let $F_u=\sum_{e \in \partial(u)} f_{e}p_{e}$. In Algorithm \ref{alg: the first attenu}, we have that in any round $t\in [n]$, each edge $e$ is probed with probability \emph{equal} to $\alp g_e$ when conditioned on its arrival.  Consider the case when $u$ is safe at time $t$. We have that $u$ is matched by the end of round $t$ with probability exactly equal to 
	\[ \sum_{e=(u,v) \in \partial(u)} \frac{r_v}{n}\alp g_e p_{e}=\sum_{e=(u,v) \in \partial(u)} \frac{r_v}{n}\frac{\alp f_e}{r_v} p_{e} =\alpha F_u/n.\]
Therefore $u$ will be safe after any time $t' \in [n]$ with probability equal to $(1-\alpha F_u/n)^{t'-1}$. Thus we have, 
\begin{align*}
\Pr[e \mbox{ is probed}]&=\sum_{t=1}^{n} \frac{r_v}{n} \alpha g_{e} \left(1-\frac{\alpha F_u}{n}\right)^{t-1} \\
&=\sum_{t=1}^{n} \frac{r_v}{n}  \frac{\alp f_e}{r_v} \left(1-\frac{\alpha F_u}{n}\right)^{t-1} \\
& \ge f_e \Big(1- \left(1-\frac{\alpha}{n}\right)^n\Big)
>  f_{e}(1-e^{-\alpha}) 
\end{align*}
After incorporating the simulation error (see Section~\ref{app:att1} in Appendix for details), we get a ratio of $1-(1-\alpha/n)^n-\epsilon$ for any given $\epsilon$.  Thus by setting $\epsilon=e^{-\alpha}-(1-\alpha/n)^n  =O(1/n)$, we get the result in theorem \ref{thm:attn1}.
\end{proof} }

Notice that $\BB_{UR}$ satisfies \textbf{Property A} with $\alpha =1/2$.
Plugging those values into the above theorem, we get a ratio of $0.3934$.

\begin{corollary} \label{thm:OC-mainproblemfr1}

When combined with $\BB_{UR}$, Algorithm \ref{alg: the first attenu} yields a competitive ratio of $0.3934$ for the Online Stochastic Matching with Timeouts problem.
\end{corollary}

Although this approach does not give our best result, we note that it places fewer restrictions on the black box than the other approaches presented in this paper. Thus, developing a stronger black box satisfying only \textbf{Property A} could potentially lead to the edge-attenuation framework yielding a better result for this problem in the future.

%%%%%----- VERTEX ATTENUATION -----%%%%%

\xhdr{Vertex-attenuation.} 
Applying vertex-attenuation without any edge-attenuation requires an offline black box $\BB$ satisfying our most strict property, \textbf{Property C}. We now motivate the intuition for the vertex-attenuation framework. Over time, the offline vertices in $U$ will be matched and therefore removed from the graph. Suppose we apply $\BB[\g_{t,v}]$ to $G_t(v)$ on each round $t$ when $v$ arrives. Thus when $t$ gets larger and $G_t(v)$ gets smaller, $\lam(e,\g_{t,v})$ will decrease for each safe edge $e=(u,v)$. This means the lower bound on probing an edge, $g_e \RR_{\BB}[\lam(e,\g_{t,v})]$, will increase with time. The trivial lower-bound is given by $\lam(e,\g_{t,v}) \leq 1$. Vertex-attenuation helps take advantage of the fact that $\lam(e,\g_{t,v})$ is a decreasing function by guaranteeing that every offline node is safe with a uniformly decreasing probability at the start of each round $t$.

Consider a specific round $t$ when $v$ comes. Let $\SC_{u,t}$ be the event that $u$ is safe at $t$ for each $u \in U$. We have the following Lemma \ref{lem:attn2}. 

%the start of update blue box	
\bluee{\begin{lemma}\label{lem:attn2}
Suppose we apply $\BB[\g_{v,t}]$ to $G_t(v)$ during each round $t$ when $v$ arrives where $\BB$ satisfies \textbf{Property C}. Then for each round $t\in [n]$ and $u,u' \in U$, we have (1) $\Pr[\SC_{u,t}] \ge (1-1/n)^{t-1}$ and 
(2) $\Pr[\SC_{u,t}  \bigwedge \SC_{u',t}] \le \Pr[\SC_{u,t}]  \Pr[\SC_{u',t}]$.
\end{lemma}

\begin{proof}
First, we show the proof of Inequality (1). Assume that $u$ is safe at the beginning of time $t-1$, \ie $\SC_{u,t-1}$. In round $t-1$, every edge $e \in \partial(u)$ will be matched with probability at most $f_e p_e/n$ by \textbf{Property C} (since $e$ is probed with probability at most $g_e$ conditioned on its arrival). Therefore, we have that $u$ is matched during round $t-1$ with probability at most $F_u/n$, where $F_u=\sum_{e \in \partial(u)} f_e p_e$. Thus we have
	\[\Pr[\SC_{u,t}] \ge \Pr[\SC_{u,t-1}] (1-F_u/n) \ge  \Pr[\SC_{u,t-1}] (1-1/n).\]
Using induction on $t$, we get that $\Pr[\SC_{u,t}]  \ge (1-1/n)^{t-1}$ for each $t \in [n]$. 

Second, we show the proof of Inequality (2). Consider the time $t-1$ and let $u$ and $u'$ both be safe (\ie both $\SC_{u,t-1}$ and $\SC_{u',t-1}$ occur). Assume that each edge $e$ is probed with probability $\alpha_e f_e$ with $\alpha_e \in [0,1]$ conditioning on its arrival at $t-1$. We get $\Pr[\SC_{u,t}  \bigwedge \SC_{u',t} \given{\SC_{u,t-1}, \SC_{u',t-1}}] =1-\sum_{e \in \partial(u) \cup \partial(u') } f_e p_e \alpha_e /n$ and $\Pr[\SC_{u,t}\given{\SC_{u,t-1}}] =1-\sum_{e \in \partial(u)} f_e p_e \alpha_e /n $, $\Pr[\SC_{u',t}\given{\SC_{u',t-1}}] =1-\sum_{e \in \partial(u')} f_e p_e \alpha_e /n $. Since $\partial (u)$ and $\partial (u')$ are disjoint, we have
	\[\Pr[\SC_{u,t}  \wedge \SC_{u',t} \given{\SC_{u,t-1}, \SC_{u',t-1}}] \le\Pr[\SC_{u,t}\given{\SC_{u,t-1}}]  \Pr[\SC_{u',t}\given{\SC_{u',t-1}}]\]
Thus, 
\begin{align*}
\Pr[\SC_{u,t} \wedge \SC_{u',t}] &=\left(\prod_{\ell=2}^t\Pr[\SC_{u,\ell}  \wedge \SC_{u',\ell} \given{\SC_{u,\ell-1}, \SC_{u',\ell-1}}] \right)\Pr[\SC_{u,1} \wedge \SC_{u',1}]\\
& \le  \left(\prod_{\ell=2}^t \Pr[\SC_{u,\ell}\given{\SC_{u,\ell-1}}]  \Pr[\SC_{u',\ell}\given{\SC_{u',\ell-1}}] \right) \Pr[\SC_{u,1}]\Pr[\SC_{u',1}]\\
&=\Pr[\SC_{u,t}]\Pr[\SC_{u',t}]
\end{align*}
\end{proof} 

Let $\beta_{u,t}=\Pr[\SC_{u,t}]$ be the probability that offline vertex $u$ is safe at $t$ before any attenuation. (As discussed before, values $\{\beta_{u,t}\}$ can obtained through simulation with sufficiently high precision.) From Lemma~\ref{lem:attn2}, we get $\beta_{u,t}\ge (1-1/n)^{t-1}$ for each $t$ and $u$. By applying an attenuation factor of $\frac{(1-1/n)^{t-1}}{\beta_{u,t}}$, we can ensure that every $u$ is safe at time $t$ with probability \textit{equal to} $(1-1/n)^{t-1}$ thus implying that $u$ is matched with probability $1-\frac{(1-1/n)^{t-1}}{\beta_{u,t}}$. Additionally, by applying vertex-attenuation \emph{independently} to every offline $u$ we get the following: events that the offline vertices being safe at each time $t$ are pairwise negatively correlated. Algorithm \ref{alg: the second attenu} gives the formal description for framework.

\begin{algorithm}[!h]
		\label{alg: the second attenu}
		\caption{$\mathsf{ATTN}_{2}[\BB]$ }
		\DontPrintSemicolon
For each $t \in [n]$, let $v$ be the vertex arriving at time $t$ and $G_t(v)$ be the star graph consisting of $v$ and its safe neighbors.\;

\bluee{Let $\beta_{u,t}$ be the probability that $u$ is safe at (the beginning of) $t$. 
Apply an attenuation factor of $\frac{(1-1/n)^{t-1}}{\beta_{u,t}}$ to each offline $u$ \emph{independently} such that each $u$ is safe at time $t$ with probability {\it equal to} $(1-1/n)^{t-1}$. 

Let $\g_{t,v}=\{f_e/r_v| e \in G_t(v)\}$ be the induced feasible solution to LP  \eqref{eqn:LP-off} from an optimal solution $\ff$ to LP \eqref{lp2:stoch-match}. Let $\BB[\g_{t,v}]$ be the feasible probing strategy of $\BB$}.\;
Apply $\BB[\g_{t,v}]$ to $G_t(v)$.
	\end{algorithm}	
\begin{theorem}\label{thm:attn2}	
For any $\BB$ satisfying \textbf{Property C} with function $\mathsf{R}_{\BB}$, $\mathsf{ATTN}_2[\BB]$ has a competitive ratio of $\int_0^1 e^{-x}\mathsf{R}_{\BB}[e^{-x}] dx-\epsilon$.
\end{theorem}	

\begin{proof}
Consider a given edge $e=(u,v)$. Let $\SC_{u,t}$ be the event that $u$ is safe at (the beginning of) $t$, i.e., $e \in G_t(v)$. From $\ATTN_2$ and Lemma~\ref{lem:attn2}, we have that
$\Pr[\SC_{u,t}]=(1-1/n)^{t-1}$ and $\Pr[\SC_{u',t}\given{\SC_{u,t}}] \le (1-1/n)^{t-1}$. Let $v$ be the vertex arriving at $t$. Recall that $\lam(e, \g_{t,v})$ is the sum of $g_{e'} p_{e'}$ over all edges in $G_t(v)$ (consisting of all safe edges incident to $v$ at $t$) excluding $e$ itself. Let $N_v$ be the set of neighbors of $v$ in the original graph $G$. Therefore, we have 
\begin{align}
\E[\lam(e, \g_{t,v})\given{\SC_{u,t}}] =\sum_{u' \in N_v, u' \neq u} (f_{u',v}/r_v) p_{u',v}\Pr[\SC_{u',t} \given{\SC_{u,t}} ]\le (1-1/n)^{t-1} \label{ineq:bb-2-1}
\end{align}
Inequality \eqref{ineq:bb-2-1} uses the fact that $\sum_{u' \in N_v, u' \neq u}f_{u',v} p_{u',v} \le t_v$. This is guaranteed by the Constraint~\eqref{lp-cons:vmatch} in LP~\eqref{lp2:stoch-match}.

  Let $\A_{e,t}$ be the event that $e$ is {\it effectively} probed during round $t$, i.e., $v$ comes at $t$, $u$ is safe, and $e$ is probed. We have 
\begin{align*}
\Pr[\A_{e,t}]  &=\Pr[ \mbox{$v$ comes at $t$}] \Pr[\mbox{$u$ is safe at $t$}] \Pr[\mbox{$e$ is probed | $u$ is safe at $t$}]  \\
&\ge \frac{r_v}{n} \Pr[\SC_{u,t}] g_e \E\Big[\mathsf{R}_{\BB}[\lam(e,\g_{t,v})] |\SC_{u,t} \Big] ~~( \textbf{Property C})  \\
&\ge  \frac{f_e}{n}(1-1/n)^{t-1} \mathsf{R}_{\BB} [\E[\lam(e,\g_{t,v})|\SC_{u,t}]] ~~ \mbox{(Convexity of $\mathsf{R}_{\BB}$, $g_e=f_e/r_v$)}\\
&\ge \frac{f_e}{n}(1-1/n)^{t-1}  \mathsf{R}_{\BB} [ (1-1/n)^{t-1}] ~~\mbox{(Inequality~\eqref{ineq:bb-2-1} and $\mathsf{R}_{\BB}$ is non-incereasing)}
 \end{align*}

Thus, we have,
\begin{align*}
 \Pr[\mbox{$e$ is probed}] &=\sum_{t=1}^n \Pr[\A_{e,t}] \\
& \ge \sum_{t=1}^n \frac{f_e}{n}\left(1-\frac{1}{n}\right)^{t-1}  \mathsf{R}_{\BB} \left[ \left(1-\frac{1}{n}\right)^{t-1}\right] =f_e\int_0^1 e^{-x}  \mathsf{R}_{\BB} [e^{-x}] dx
\end{align*}

The right-most equality is obtained by letting $n \rightarrow \infty$.
Incorporating simulation errors (see Section~\ref{app:att2}), we get a competitive ratio of $\int_0^1 e^{-x}  \mathsf{R}_{\BB} [e^{-x}] dx -\epsilon$, for any given $\epsilon>0$. 
\end{proof}
}	
%the end of update blue box	
	
Plugging the $\BB_{UR}$ function $\mathsf{R}_{\BB_{UR}}[x]=1-x/2$ into the above formula, we get a ratio of $0.4159$.	

\begin{corollary} \label{thm:OC-mainproblemfr2}
The second framework combined with $\BB_{UR}$ yields an algorithm, which achieves a competitive ratio of $0.4159$ for the Online Stochastic Matching with Timeouts problem.
\end{corollary}

%%%%%----- EDGE AND VERTEX ATTENUATION -----%%%%%

\xhdr{Edge and Vertex-attenuation Combined.} 
Our final and currently most powerful framework combines both edge and vertex-attenuation. Notice that by design, edge-attenuation upper bounds the probability an edge will be probed in an offline round. Therefore, our black box only needs to satisfy \textbf{Property B} which is slightly less restrictive than \textbf{Property C}.

At the start of each round $t$, let every $u$ be safe with a target probability equal to $\gamma_t \in [0,1]$. From \textbf{Property B}, we have that each safe edge $e=(u,v)$ is probed during round $t$ with probability at least $g_e \alpha_t$ (conditioning on the arrival of $v$ and $u$ is safe at $t$), where $\alpha_t=\E[\RR_{\BB}[\lam(e,\g_{t,v})]] \ge \RR_{\BB}[\gamma_t]$ (same analysis as Theorem~\ref{thm:attn2}).  Using edge-attenuation, each safe edge $e$ is probed with probability equal to $\RR_{\BB}[\gamma_t] g_e$. Consequently, each safe $u$ at time $t$ will remain safe at $t+1$ with probability at least $1- \RR_{\BB}[\gamma_t]/n$. Through vertex-attenuation, each $u$ remains safe at $t+1$ with probability equal to $\gamma_{t+1}=\gamma_t(1- \RR_{\BB}[\gamma_t]/n)$. Thus, we ensure each edge is probed with a uniformly increasing ratio and every offline node is safe with a uniformly decreasing probability. Algorithm \ref{alg: the third attenu} describes this online framework denoted $\mathsf{ATTN}_3$.

\begin{algorithm}[!h]
		\label{alg: the third attenu}
		\caption{$\mathsf{ATTN}_{3}[\BB]$ }
		\DontPrintSemicolon
	 For times $1, 2, \ldots, t$ do\;
	 
	  Let each $u$ be safe with probability equal to $\gamma_t$. \;
 
		Let $v$ arrive at time $t$ and $G_t(v)$ be the graph of $v$ and its safe neighbors. \bluee{Let $\g_{t,v}=\{f_e/r_v| e \in G_t(v)\}$ be an induced feasible solution to LP \eqref{eqn:LP-off} from an optimal solution $\ff$ to LP \eqref{lp2:stoch-match}}. \;
	
		Apply $\BB[\g_{t,v}]$ and edge-attenuation to $G_t(v)$ such that each edge $e$ is probed with probability equal to $\alpha_t g_e$, where  $\alpha_t=\RR_{\BB}[\gamma_t]$. \;

		Apply vertex-attenuation to each $u$ such that each $u$ is safe at time $t+1$ with probability equal to $\gamma_{t+1}=\gamma_t (1-\alpha_t/n)$.

	\end{algorithm}	
	
	We can express a recurrence relation for $(\gamma_t, \alpha_t)$ as follows.
\begin{equation} \label{eqn:att3}
\textstyle \gamma_1=1,\quad 
\alpha_t=\RR_{\BB}[\gamma_t]; \quad
\gamma_{t+1}=\gamma_t\left(1-\frac{\alpha_t}{n}\right) 
\end{equation}

\begin{theorem}	\label{thm:attn3}
For any $\BB$ satisfying \textbf{Property B}, $\mathsf{ATTN}_3[\BB]$ has an online competitive ratio of $ (1-h(1)-\epsilon)$, where $h$ is the unique function satisfying $h'=-h\RR_{\BB}[h]$ with boundary condition $h(0)=1$. Here, $h'$ represents the first-order derivative of function h.
\end{theorem}	
\begin{proof}
Consider an edge $e=(u,v)$. It will be probed with probability equal to $f_e\sum_{t=1}^n \frac{\gamma_t \alpha_t}{n}$. 
From Observation \ref{obs:pro}, $\RR_{\BB}[x] \in [0,1]$ for all $x\in [0,1]$. 
From \textbf{Property B} and Equation \eqref{eqn:att3}, we know that $\{\alpha_t \}$ is an increasing sequence and $\{\gamma_t\}$ is a decreasing sequence with $\gamma_t \ge 1/e $ and $\alpha_t \le \RR_{\BB}[1/e]$ for all $t$.

Define a function $h: [0,1] \rightarrow [0,1]$ such that $h((t-1)/n)=\gamma_t$ for all $t\in [n]$. Thus we have $h(0)=1$. Equation \eqref{eqn:att3} implies that 
$$\frac{h(t/n)-h((t-1)/n)}{1/n}=-h((t-1)/n) \RR_{\BB}[h((t-1)/n)]$$ 
Letting $x=(t-1)/n$ and the above equation yields $\frac{h(x+1/n)-h(x)}{1/n}=-h(x) \RR_{\BB}[h(x)]$. Now letting $n \rightarrow \infty$, we can see that $h$ satisfies the differential equation $h'=-h \RR_{\BB}[h]$ with boundary condition $h(0)=1$. 

Given $h$, we have
\begin{align*}
\sum_{t=1}^n \frac{\alpha_t \gamma_t}{n}& =\frac{1}{n} \sum_{t=1}^n h((t-1)/n) \RR_{\BB}[h((t-1)/n)] \\
&=\int_0^1 h(x) \RR_{\BB}[h(x)] dx=h(0)-h(1)=1-h(1)
\end{align*}
Simulation error subtracts at most $O(\epsilon)$ in the final ratio (see Section~\ref{app:att3}). Hence, this completes the proof of the theorem.
\end{proof}
$\BB_{UR}$ satisfies \textbf{Property B} with $\RR_{\BB_{UR}}[x]=1-x/2$. Plugging $\RR_{\BB_{UR}}$ into the above theorem, we get $h(x)=2/(1+e^x)$, which implies $\mathsf{ATTN}_3[\BB_{UR}]$ has an online ratio of $1-h(1)-\epsilon \ge  0.4621$. 
\begin{corollary} \label{thm:OC-mainproblem}
The Attenuation framework	$\mathsf{ATTN}_3[\BB]$ combined with $\BB_{UR}$ yields an algorithm, which achieves a competitive ratio of $0.4621$ for the Online Stochastic Matching with Timeouts problem.
\end{corollary}

%~\\~\\~\\~\\~\\
%
%
%%%%%----- BELOW IS OLD SECTION -----%%%%%
%
%

%\subsection{Older version section 3 below}

%	In this section, we present three online attenuation frameworks. Each framework invokes an offline black box with certain required properties. We first present the three attenuation frameworks in sections \ref{sec:att1}, \ref{sec:att2}, \ref{sec:att3}, respectively. Then in section \ref{sec:extension}, we offer an example of how the first attenuation framework can be extended to the more general model, Online Stochastic Matching with two-sided timeouts. Throughout this section, we assume that through simulations we can always get an accurate estimation of our target probabilities. Just as shown in \cite{AGM, ma2014improvements}, we can manipulate the simulation errors properly such that we lose at most an additive factor of $\epsilon$ in the final ratio. 
%	%We will discuss how to remove this assumption in the Appendix. 
%	
%	Our first framework (an improvement on the approach in~\cite{AGM}) upper bounds the probability that each edge will be safely probed. This balances the performance of all edges and improves the worst case competitive ratio. Our second framework is a novel approach that upper bounds the probability that a vertex in $U$ will be safe at time $t$. This lets us to exploit the fact that the star graphs $G(v)$ will be smaller in later rounds leading to a higher probability of safely probing each of their edges. Our third framework combines these two approaches.

\subsection{Extension to the Two-sided Timeouts}	
\label{sec:extension}
The online attenuation framework combined with an offline black box can be extended to more general models.
In this section, we give an example by showing how the first attenuation framework together with an offline black box $\BB$ satisfying \textbf{Property A} can be used for the generalization of Stochastic Matching with timeouts on both offline and online vertices. We believe the other two frameworks can be used to attack the generalized model as well.

In this model, each offline vertex $u$ has a finite timeout constraint of $t_u <n$, \ie  each $u$ can be probed at most $t_u$ times over the $n$ rounds in addition to the constraints of our previous setting. Hence, the constraint \ref{lp-cons:upatience} in LP (\ref{lp2:stoch-match}) is a valid constraint in the benchmark. 

\begin{theorem} \label{thm:ext}
For any $\BB$ satisfying \textbf{Property A} with $\alpha$, $\mathsf{ATTN}_1[\BB]$ yields an online competitive ratio of $\alpha e^{-\alpha}-\epsilon$ for the Online Stochastic Matching with Two-sided Timeouts problem.
\end{theorem}

Recall that $\SC_{u,t}$ is the probability that $u$ is safe at time $t$. In this new setting, $u$ is safe if $u$ is not matched and the timeout of $u$ has not been exhausted. The lemma \ref{lemma:scut} gives a lower bound on $\Pr[\SC_{u,t}]$ when we apply $\ATTN_1[\BB]$ using any $\BB$ satisfying \textbf{Property A} with $\alpha$.

\begin{lemma}
	\label{lemma:scut}
We have the following bound on the probability of the safe event.
\[ \Pr[\SC_{u,t}] \ge \left(1-\frac{\alpha}{n}\right)^{t-1}\left(1-\frac{\alpha(t-1)}{n}\right).\]
\end{lemma}

\begin{proof}
Consider a given vertex $u$. For each $e \in \partial(u)$ and each $t' \in [n]$, let $X_{e,t'}$ be the indicator that $e$ comes at $t'$; $Y_{e,t'}$ be the indicator that $e$ is probed when $e$ comes at $t'$ and $u$ is safe at $t'$; $Z_{e,t'}$ be the indicator  that $e$ is present when probed. \bluee{The random variables $\{X_{e, t'}, Y_{e,t'}, Z_{e,t'}\}$ are all independent for any given $(t',e)$. In particular, for each $e=(u,v) \in \partial(u)$ and $t'\in[n]$, we have that 
\[	\Pr[X_{e,t'}=1]=r_v/n, \; \; \Pr[Y_{e,t'}=1]=\alp g_e=\alp f_e/r_v, \; \; \Pr[Z_{e,t'}]=p_e \]
The second equality is due to $\ATTN_1[\BB]$ where $\BB$ satisfies \textbf{Property A}. }

Let $\SC^{1}_{u,t}$ be the event that $u$ is not matched at time $t$ and $\SC^2_{u,t}$ be the event that $u$ is probed at most $t_u-1$ at the beginning of time $t$. Define $A_1$ as the event that $\sum_{t'=1}^{t-1} \sum_{e \in \partial(u) }X_{e,t'} Y_{e,t'} Z_{e,t'}=0$ and $A_2$ as the event that $\sum_{t'=1}^{t-1} \sum_{e \in \partial(u) }X_{e,t'} Y_{e,t'} \le t_u-1$. Observe that $\Pr[\SC^{1}_{u,t} \bigwedge \SC^{2}_{u,t}] \ge \Pr[A_1 \wedge A_2]$. Let us now lower bound the value of $\Pr[A_1 \bigwedge A_2]$. 

Recall that $F_u =\sum_{e \in \partial(u)} f_e p_e$. For each given $t'<t$, we know the following --- $\Pr[\sum_{e \in \partial(u)} X_{e, t'} Y_{e,t'} Z_{e,t'} =0]=1-\alpha F_u/n \ge 1-\alpha/n$. Therefore we have $\Pr[A_1] \ge  (1-\alpha/n)^{t-1}$. Notice that for each given $t'<t$,
\begin{align*}
 \E\left[\sum_{e \in \partial(u)} X_{e,t'} Y_{e,t'}\given{A_1}\right] 
& = \E\left[\sum_{e \in \partial(u)} X_{e,t'} Y_{e,t'} \given{\sum_{e \in \partial(u)} X_{e, t'} Y_{e,t'} Z_{e,t'} =0}\right] \\
& =  \sum_{e\in \partial(u)}\Pr\left[ X_{e,t'}=Y_{e,t'}=1\given{\sum_{e \in \partial(u)} X_{e, t'} Y_{e,t'} Z_{e,t'} =0}\right] \\
&= \sum_{e\in \partial(u)} \frac{\Pr\left[ X_{e,t'}=Y_{e,t'}=1,  \sum_{e \in \partial(u)} X_{e, t'} Y_{e,t'} Z_{e,t'} =0 \right]}{\Pr[\sum_{e \in \partial(u)} X_{e, t'} Y_{e,t'} Z_{e,t'} =0]} \\
&=\sum_{e\in \partial(u)} \frac{\Pr\left[ X_{e,t'}=Y_{e,t'}=1,  Z_{e,t'}=0 \right]}{\Pr[\sum_{e \in \partial(u)} X_{e, t'} Y_{e,t'} Z_{e,t'} =0]}\\
&=\sum_{e\in \partial(u)}  \frac{(r_v/n) (\alpha g_e) (1-p_e)}{1- \alpha F_u/n}
 = \sum_{e\in \partial(u)}  \frac{\alpha f_e (1-p_e)/n}{1- \alpha F_u/n}\\
& \le  \frac{\alpha (t_u-F_u) /n}{1- \alpha F_u/n} \le \frac{\alpha t_u}{n}
\end{align*}

Thus we get $\E[\sum_{t'=1}^{t-1}\sum_{e \in \partial(u)} X_{e,t'} Y_{e,t'}\given{A_1}] \le
 \frac{\alpha t_u}{n} (t-1)$, which implies that $\Pr[A_2\given{A_1}] \ge \left(1-\frac{\alpha}{n}(t-1)\right)$ by Markov's Inequality. Therefore 
\[\Pr[\SC_{u,t}]=\Pr\left[\SC^1_{u,t} \bigwedge \SC^2_{u,t}\right] \ge \Pr\left[A_1 \bigwedge A_2\right] \ge \left(1-\frac{\alpha}{n}\right)^{t-1}\left(1-\frac{\alpha (t-1)}{n}\right) \qedhere \]
\end{proof}

Let us now prove Theorem \ref{thm:ext}. 

\begin{proof}
The proof is very similar to that of Theorem \ref{thm:attn1}. Consider a single edge $e=(u,v)$, we have \bluee{
\[  \Pr[e \mbox{ is probed}] \ge \sum_{t=1}^{n} \frac{r_v}{n} \alpha g_{e} \left(1-\frac{\alpha}{n}\right)^{t-1} \left(1-\frac{\alpha(t-1)}{n}\right) \ge  f_{e} \alpha e^{-\alpha} \qedhere
\] }
\end{proof}

Plugging $\BB_{UR}$ with $\alpha=0.5$, we get a ratio of $0.303$ for this generalized model. 
\begin{corollary} \label{thm:OC-bothsides}
	The first framework $\mathsf{ATTN}_1[\BB]$ combined with $\BB_{UR}$ yields an algorithm, which achieves a competitive ratio of $0.303$ for the Online Stochastic Matching with Two-sided Timeouts problem.
\end{corollary}

\section{Lower Bound to the Benchmark LP}
\label{sec:lb}

Here we present an unconditional lower bound for this LP due to the stochasticity of the problem. We call this lower bound a \textit{stochasticity gap}, similar to the concept of an integrality gap.
	
	 Consider a complete bipartite graph with $|U| = |V| = n$. Let the edge probabilities $p_e$ for all edges be $1/n$ and the rewards $w_e$ be $1$. Let the patience values for all vertices be $n$. Notice that assigning $f_e = 1$ for every edge is a feasible solution to LP \eqref{lp2:stoch-match}. Hence, the optimal LP value is at least $n$. However, we will show that any online algorithm cannot perform better than $(1-1/e)n$. Therefore, the stochasticity gap for this LP is at least $(1-1/e) \approxeq 0.63$.
	 
	 Consider any vertex $u \in U$. We have the following:
	 \begin{align}
		 \Pr[\text{$u$ is matched}]
		 & =1-\Pr\left[\bigwedge_{t=1}^{n} \text{$u$ is not matched at $t$}\right]&\\
		 & =1-\prod_{t=1}^{n}\Pr[\text{u was not matched at t}]  \label{lb:cons2}\\
		 &\leq 1 - \prod_{t=1}^{n}\left( 1- \frac{1}{n} \times \frac{1}{n} \times n \right)\label{lb:cons3}\\
		 &\leq 1 - 1/e -o(1)
	 \end{align}
	 
	 Equality \eqref{lb:cons2} is due to independence. Inequality \eqref{lb:cons3} uses union bound, and the facts that (1) $p_e=1/n$ and (2) with probability $1/n$ each $v$ is drawn in each round. By applying linearity of expectation, we claim that no algorithm can do better than $(1-1/e)n$.

\section{Conclusion and Future Directions}
	We gave a general framework for the Online Stochastic Matching with Timeouts problem and its extension. This led to improved competitive ratios for the former and first constant factor ratio for the latter. More importantly, the frameworks are general enough to obtain further improvements by simply finding a better black box for the offline problem on star graphs.  %Hence, future improvements for the online problem can be obtained by improving algorithms for a specialized offline problem. 
	One future direction is to increase the competitive ratio by designing better black boxes. Another future direction is to design similar framework(s) for the various other online stochastic matching problems, such as $b$-matching. We believe these frameworks have the potential to give a unified framework for many of the stochastic matching problems.

\bibliographystyle{acm}
\bibliography{refs}

\appendix

\section{Appendix}

\subsection{Error Accumulation in the First Attenuation Framework}
\label{app:att1}

Our simulation-based edge-attenuation approach is very similar to that shown in \cite{AGMesa}. For more details, please refer to the Appendix B in \cite{AGMesa}. Here we assume that with probability $(1-\epsilon)$, we can similarly obtain that all safe edges with $f_e \ge \epsilon/n$ should be probed with probability $[\alpha f_e/(1+\epsilon), \alpha f_e/(1-\epsilon)]$ in all rounds. For those edges with $f_e <\epsilon/n$, we add no attenuation.

Now we show that the error from simulation accumulates at most $O(1)$ times in the final ratio. Consider an edge $e=(u,v)$. From the analysis in Theorem \ref{thm:attn1}, we have that in each round, $u$ will be matched with probability at most $F'_u=\frac{1}{n}\sum_{e \sim \partial(u)} \frac{\alpha f_e p_e}{1-\epsilon}+\frac{\epsilon}{n} \le \frac{\alpha}{n}+\frac{2\epsilon}{n}$, if $u$ is safe.  Therefore $e$ will be probed with probability at least 
$$\Pr[e \mbox{ is probed}]\ge(1-\epsilon) \sum_{t=1}^{n} \frac{1}{n} \frac{\alpha f_{e}}{1+\epsilon} \left(1-\frac{\alpha+2\epsilon}{n}\right)^{t-1} \ge f_e \Big(1-(1-\alpha/n)^n-O(\epsilon)\Big) $$

\subsection{Simulation-Based Vertex-Attenuation in the Second Attenuation Framework}
\label{app:att2}

%In the first attenuation framework, the proper attenuation factors for each edges can be computed in each round independently. However, 

The vertex-attenuation approach is slightly more involved compared to edge-attenuation. Consider a node $u$ and let $\beta_t$ and $\beta'_t$ be the probability that $u$ is safe at time $t$ before and after attenuation, respectively. Define the event  $E_t$ as $\beta'_t \in \left[ (1-1/n)^{t-1}\frac{1}{1+\epsilon}, (1-1/n)^{t-1}\frac{1}{1-\epsilon} \right]$.
Here we show how to achieve the goal that $\left( \bigwedge_{t\in [n]} E_t \right)$ occurs with probability at least $1-\epsilon$.

For $t=1$, we do not need any attenuation. From Lemma \ref{lem:attn2}, we have that
$\beta_2 \ge (1-1/n)$. Let $\hat{\beta}_2$ be the estimation obtained from $N$ experiments. Thus we see
$$\Pr\left[|\hat{\beta}_2-\beta_2| \ge \epsilon \beta_2 \right] \le 2 \exp \left( -\frac{\epsilon^2}{3 e} N \right)$$

Thus by setting $N=\ln(2/\epsilon) (3e/\epsilon^2 )$, we claim that with probability at least $1-\epsilon$, $\hat{\beta}_2 \in [\beta_2 (1-\epsilon), \beta_2 (1+\epsilon)]$. The resulting attenuation factor is defined as follows: $\sigma_2=1$ if $\hat{\beta}_2 <(1-1/n)$ and $\sigma_2=\frac{1-1/n}{\hat{\beta}_2}$ if  $\hat{\beta}_2 \ge (1-1/n)$. At the beginning of $t=2$, we keep $u$ with probability $\sigma_2$ and throw away $u$ otherwise, and do this independently for other LHS nodes. 

 Now assume $\hat{\beta}_2 \in \left[ \beta_2 (1-\epsilon), \beta_2 (1+\epsilon) \right]$ occurs. Note that $\beta'_2=\beta_2*\sigma_2$. We have two cases.

\begin{itemize}
\item $\hat{\beta}_2 <(1-1/n)$. In this case, we have $\beta_2<(1-1/n)\frac{1}{1-\epsilon}$. And hence, $\beta'_2=\beta_2 \in \left[ (1-1/n), (1-1/n)\frac{1}{1-\epsilon} \right]$. 

\item $\hat{\beta}_2 \ge (1-1/n)$. In this case we have $\beta'_2=\beta_2 \frac{1-1/n}{\hat{\beta}_2} \in \left[ (1-1/n)\frac{1}{1+\epsilon}, (1-1/n)\frac{1}{1-\epsilon} \right]$.

\end{itemize}

Therefore at $t=2$ with probability $1-\epsilon$, the event $E_2$ occurs. Now assume $E_2$ occurs. From Lemma \ref{lem:attn2}, we see
$\beta_3 \ge (1-1/n) \beta_2' \ge (1-1/n)^2\frac{1}{1+\epsilon}$. Using the same analysis as above and the same value of $N$, we have that with probability $1-\epsilon$, the event $E_3$ occurs. Similarly the analysis carries through from $E_3$, $E_4$, $\ldots$,  $E_n$. Therefore we claim that $\Pr[ \bigwedge_{t\in [n]} E_t] \ge (1-\epsilon)^n \ge 1- n\epsilon$. Finally, in each round we scale down the error probability by a factor of $1/n$.
\\

%with probability at least $(1-\epsilon)$, $\forall u \in U, \forall t\in [n]$, probability that $u$ is safe $\in [(1-1/n)^{t-1}(1-\epsilon), (1-1/n)^{t-1}(1+\epsilon)]$

\yhdr{Error Accumulation in the Second Attenuation Framework}
\\

From the analysis above, we can safely assume that in the second attenuation framework, 
 with probability at least $(1-\epsilon)$,
 $\Pr[\mathcal{S}_{u,t}] \in [(1-1/n)^{t-1}(1-\epsilon), (1-1/n)^{t-1}(1+\epsilon)] $ for all $u$ and $t$. We can achieve this by setting the simulation number in each round as $N=\ln(2nm/\epsilon) (3e /\epsilon^2 )$.  Condition on this, we show how the error accumulates in the final ratio. Consider a given edge $e=(u,v)$. Notice that
 $\E[R_{e,\ff_{t,v}}\given{\SC_{u,t}}] \le (1-1/n)^{t-1}(1+\epsilon)$, which implies that
 $$\E[\RR_{\BB}[R_{e,\ff_{t,v}}]] \ge \RR_{\BB}[(1-1/n)^{t-1}(1+\epsilon)] \ge 
  \RR_{\BB}\left[(1-1/n)^{t-1}\right]- \epsilon M$$
 where $M$ is a constant upper bound for absolute value of the first derivative of $\RR_{\BB}$ over $[0,1]$. Recall that $\A_{e,t}$ is the event that $e$ is {\it effectively} probed during round $t$. Applying the same analysis in Theorem \ref{thm:attn2}, we have
\begin{eqnarray*} 
\Pr[\mbox{e is probed}] &\ge&  (1-\epsilon)\sum_{t=1}^n \Pr[\A_{e,t}] \\
&\ge& 
 (1-\epsilon)^2 \sum_{t=1}^n \frac{f_e}{n}\left(1-\frac{1}{n}\right)^{t-1} \left( \mathsf{R}_{\BB} \left[ \left(1-\frac{1}{n}\right)^{t-1}\right] -\epsilon M \right)=\int_0^1 e^{-x}  \mathsf{R}_{\BB} [e^{-x}] dx
-O(\epsilon)
\end{eqnarray*} 
 
Therefore by setting $\epsilon$ small enough, we can get a ratio of  $\int_0^1 e^{-x}  \mathsf{R}_{\BB} [e^{-x}] dx
-\epsilon$ for any given $\epsilon>0$.

%%%%%%%%%
 \subsection{Simulation-Based Attenuation in the Third Attenuation Framework}
\label{app:att3}
 
For the third attenuation framework, we need the following key ingredient. Suppose we have a random variable $X$ with $\E[X]=\mu \in [\beta-\epsilon,1]$ where $0<\beta<1$. The random variable models the event that an edge $e$ is probed in some round or a LHS node is safe at some round $t$. Through analytical analysis, we know a good lower bound $\beta$ for the unknown mean value $\mu$ with error $\epsilon$. Now we need to compute a proper attenuation factor $\sigma \in [0,1]$ such that $\sigma \mu$ is very close to $\beta$ with high probability.  Consider the following simulation-based approach: we sample the random variable $X$ for $N$ times and let $\hat{\mu}$ be the sample mean; define $\sigma=\beta/\hat{\mu}$ if $\hat{\mu} \ge \beta$ and  $\sigma=1$ otherwise. Assume $\mu \ge \beta-\epsilon \ge \beta/2$. 

\begin{lemma} \label{lem:appn-attn3}
When $N= \frac{6}{\epsilon^2 \beta} \ln \frac{2}{\delta} $, we have that with probability at least $1-\delta$, $\sigma \mu \in [\beta-\epsilon, \frac{\beta}{1-\epsilon}]$
\end{lemma}
\begin{proof}
By applying Chernoff Bound, we see
$$\Pr[|\hat{\mu}-\mu| \ge \epsilon \mu ] \le 2 \exp \left( -\frac{\epsilon^2}{3} N \mu\right) \le 2 \exp \left( -\frac{\epsilon^2}{3} \frac{\beta N}{2} \right) = \delta$$

Thus with probability $1-\delta$, $\hat{\mu} \in [(1-\epsilon)\mu, (1+\epsilon)\mu]$. Assume this occurs. Consider the first case $\hat{\mu} \ge \beta$. Then $\sigma \mu=\beta \frac{\mu}{\hat{\mu}} \in [\beta/(1+\epsilon), \beta/(1-\epsilon)]$. We are done since $\beta/(1+\epsilon) \ge \beta-\epsilon$.  For the second case
$\hat{\mu} < \beta$, we see that $\mu \le \beta/(1-\epsilon)$. Thus
we have $\sigma\mu=\mu \in [\beta-\epsilon, \beta/(1-\epsilon)]$.

\end{proof}

WLOG assume all $f_e \ge \epsilon/n$ and $\RR_{\BB}$ have finite first derivative and upper bounded by $1/2$. Recall that in Equation \eqref{eqn:att3}, $\alpha_t \ge \alpha_1>0$ and $\gamma_t \in [1/e, 1]$ for all $t\in [n]$.

Consider the first round $t=1$.  Consider an edge $e$ and let $\beta'_{e,1}$ be the probability that $e$ is probed during the round $t=1$ before attenuation. Through the analysis in Section~\ref{sec:attenuation}, we see $\beta'_{e,1} \ge \alpha_1 f_e$. 
Let $\beta''_{e,1}$ be the probability that $e$ is probed during the round $t=1$ after attenuation. 
From Lemma \ref{lem:appn-attn3} and by setting $N=O(\ln (1/\delta)n/\epsilon^3)$, we have that with probability $1-\delta$, $\beta''_{e,1} \in [\alpha_1 f_e-\epsilon,\alpha_1 f_e/(1-\epsilon) ]$. Let $A_1$ be the event that during round $t=1$,  $\beta''_{e,1} \in[\alpha_1 f_e-\epsilon,\alpha_1 f_e/(1-\epsilon) ]$ for all $e$. By union bound and by setting $N=O(\ln (nm/\delta)n/\epsilon^3)$ (where $m=|U|$ and $|E| \le mn$), we can ensure $A_1$ occurs with probability $1-\delta$.

Now condition on $A_1$. Consider a node $u$ at the beginning of $t=2$. Let $\gamma'_{u,2}$ and $\gamma''_{u,2}$ be the probability that $u$ is safe at the beginning of $t=2$ before and after attenuation. We have $\gamma'_{u,2} \ge 1-\frac{\alpha_1 /(1-\epsilon)}{n} \ge \gamma_2-\epsilon_2$ where $\epsilon_2\doteq 2\epsilon/n$. Here WLOG assume $\epsilon<1/2$. Let $B_2$ be the event that $\gamma''_{u,2}\in [\gamma_2-\epsilon_2, \gamma_2/(1-\epsilon_2)]$ for all $u \in U$. Similarly we can ensure $B_2$ occurs with probability $1-\delta$ by setting $N=O(\ln (m/\delta)n^2/\epsilon^2)$.

%after attenuation all LHS nodes are safe at the time $t=2$ with probability within $[\gamma_2-\epsilon_2, \gamma_2/(1-\epsilon_2)]$. From the previous analysis and Lemma~\ref{lem:appn-attn3}, we can make it that $B_2$ occurs with probability $1-\delta$ by setting $N=O(\ln (m/\delta)n^2/\epsilon^2)$.

Now condition on both $A_1$ and $B_2$. For a given safe edge $e$, let $\beta'_{e,2}$ be the probability that $e$ is probed during the round $t=2$ before attenuation. From previous analysis, we have 
$$\beta'_{e,2} \ge f_e \mathsf{R}_{\BB}[\gamma_2''] \ge f_e (\mathsf{R}_{\BB}[\gamma_2+2\epsilon_2] ) \ge f_e \alpha_2-\epsilon_2 $$

Let $\beta''_{e,2}$ be the probability that $e$ is probed during the round $t=2$ after attenuation and $A_2$ is the event that $\beta''_{e,2}\in [f_e \alpha_2-\epsilon_2, f_e \alpha_2/(1-\epsilon_2)]$ for all safe $e$ at $t=2$.
Applying Lemma \ref{lem:appn-attn3} and union bound, we can make sure that $A_2$ occurs with probability $1-\delta$, by setting $N=O(\ln (mn/\delta) n^3/\epsilon^3)$. 

Now condition on $A_1, B_2, A_2$. Similarly, let $\gamma'_{u,3}$ and $\gamma''_{u,3}$ be the probability that $u$ is safe at time $t=3$ before and after attenuation. Note that 
$$\gamma'_{3,u} \ge \gamma''_{u,2}\left(1-\frac{\alpha_2/(1-\epsilon_2)}{n} \right) \ge \left(\gamma_2-\epsilon_2\right)\left(1-\frac{\alpha_2}{n}-\frac{2\epsilon_2}{n}\right) \ge \gamma_3-\epsilon_2 \left(1+\frac{2}{n}\right) \doteq \gamma_3-\epsilon_3$$

Define $B_3$ as the event that $\gamma''_{u,3}\in [\gamma_3-\epsilon_3, \gamma_3/(1-\epsilon_3)]$ for all $u \in U$. Similarly we can ensure $B_3$ occurs with probability $1-\delta$ by setting $N=O(\ln (m/\delta)n^2/\epsilon^2)$.

 Let $\gamma''_{u,t}$ be the probability that $u$ is safe at time $t$ and $\beta''_{e,t}$ the probability that $e$ is probed during round $t$ when it is safe, after attenuation. Define $\epsilon_t=\frac{2\epsilon}{n}\left(1+\frac{2}{n}\right)^{t-2}$ for each $t \ge 2$ and $\epsilon_1=\epsilon$ for $t=1$. Similarly, 
let $B_t$ with $t>1$ be the event that $\gamma''_{u,t}\in [\gamma_t-\epsilon_t,\gamma_t/(1-\epsilon_t)]$ for all $u$ and $A_t$ with $t \ge 1$ be the event that $\beta''_{e,t}\in [f_e \alpha_t-\epsilon_t,f_e \alpha_t/(1-\epsilon_t)]$ for all safe $e$. Doing a similar analysis as above we have the following two observations.

{\it
\begin{itemize} 
\item Condition on $\{A_{t'}, B_{t'} | t'<t\}$. We can ensure that $B_t$ occurs with probability $1-\delta$ by setting $N=O(\ln(m/\delta)n^2/\epsilon^2)$.

\item Condition on $\{A_{t'}, B_{t'} | t'<t\}$ and $B_t$. We can ensure that $A_t$ occurs with probability $1-\delta$ by setting $N=O(\ln(mn/\delta)n^3/\epsilon^3)$.
\end{itemize}}

Therefore by setting $\delta=\epsilon/(2n)$, we achieve that with probability $1-\epsilon$, all events in $\{A_t, B_t| t\in [n]\}$ occur. Note that during each round and for each edge or LHS node, our sampling size is $N=O(\ln (mn^2/\epsilon) n^3/\epsilon^2) $.
\\

\yhdr{Error Accumulation in the Third Attenuation Framework}
\\

Now assume all events in $\{A_t, B_t| t\in [n]\}$ occur. Consider an edge $e=(u,v)$. The performance should be at least
$$(1-\epsilon)\sum_{t=1}^n \frac{\gamma''_{u,t} \beta''_{e,t}}{n}\ge (1-\epsilon) \sum_{t=1}^n \frac{(\gamma_t-\epsilon)(\alpha_t f_e-\epsilon)}{n}=\sum_{t=1}^n \frac{\gamma_t \alpha_t f_e}{n}-O(\epsilon)$$ 

This completes the description of the error analysis.

%Focus on a single round $t\in [n]$ and graph $G_t(v)$. For each $e \in G_t(v)$, let $\beta_e$ be the probability that $e$ is probed and from assumption, we know $\beta_e \ge \alpha f_e$. We run simulations for $N$ times and get an estimation of $\hat{\beta}_e$ for $\beta_e$. By Chernoff Bound, we see for each edge $e$ with $f_e \ge \epsilon/n$, we have
%$$\Pr[|\hat{\beta}_e-\beta_e| \ge \epsilon \beta_e  ] \le 2 \exp \left( -\frac{\alpha\epsilon^3}{3n} N \right)\doteq \frac{\epsilon}{nm}$$
%
%where $m=|U|$. We can achieve the second equality by setting $N=\frac{3n^2m}{\alpha \epsilon^4}$ and by union bound, we claim that with probability at least $(1-\epsilon)$, the first inequality holds for all safe edges with $f_e \ge \epsilon/n$ in all rounds. Therefore the dumping factor for each safe edge $e$ is defined as $\sigma_e=\min (\alpha f_e/ \hat{\beta}_e,1 )$ if $f_e \ge \epsilon/n$ and $1$ otherwise. While running the algorithm $\BB[\f_{t,v}]$, we try to probe each edge $e$ with probability $\sigma_e$ when $e$ is safe and with probability $1-\sigma_e$, we simulate the probe of $e$ and stop t 

\bluee{
\section{Summary of Notation}
\label{appx:notations}
In this section, we summarize all of the notation used throughout the paper.

\begin{center}
\begin{tabular}{ |c|l| }
\hline
	Notation & Usage  \\
\hline
 $n$ & total number of online rounds \\ \hline
 $r_v$ & expected number of arrivals of vertex $v$ in the $n$ online rounds \\ \hline
 $\partial(u)$, $\partial(e)$ & the set of edges incident to vertex $u$ and edge $e$ respectively \\ \hline
 $\mathbf{f}$ & optimal solution to LP~\eqref{eqn:LP-off} \\ \hline
 $F_u$ & $\sum_{e \in \partial(u)} f_e p_e$ \\ \hline
 $\mathbf{g}$ & scaled fractional solution feasible to LP~\eqref{eqn:LP-off} \\ \hline
 $G_t(v)$ & The (random) star graph incident to $v$ with all of the \emph{safe} neighbors at time $t$ in a run of \\& the algorithm \\ \hline
 $\mathbf{\hat{G}}$ & (random) integral solution obtained by running GKPS on $\mathbf{g}$ \\ \hline
 $S_{u, t}$ & Event (random) that $u$ is safe at time $t$ in any run of the algorithm \\ \hline
 $\mathcal{A}_{e, t}$ & Event (random) that for an edge $e=(u, v)$, $v$ comes at time $t$, $u$ is safe and $e$ is probed \\ \hline
 $\lambda(e, \mathbf{g})$ & $\sum_{e' \neq e} g_{e'} p_{e'}$ \\ \hline
 
 \end{tabular}
\end{center}
}

\end{document}